\documentclass{eptcs}
\providecommand{\event}{DCM 2012} 
\usepackage{breakurl}             
\usepackage{rotating}
\usepackage{amssymb,amsmath,amsthm}

\newtheorem{thm}{Theorem}

\newtheorem{cor}[thm]{Corollary}
\newtheorem{defi}[thm]{Definition}
\newtheorem{rem}[thm]{Remark}
\newtheorem{nota}[thm]{Notation}
\newtheorem{exa}[thm]{Example}

\newtheorem{princ}[thm]{Principle}

\newcommand\be{\begin{equation}}
\newcommand\ee{\end{equation}}

\newbox\gnBoxA
\newdimen\gnCornerHgt
\setbox\gnBoxA=\hbox{$\ulcorner$}
\global\gnCornerHgt=\ht\gnBoxA
\newdimen\gnArgHgt

\def\Godelnum #1{%
	\setbox\gnBoxA=\hbox{$#1$}%
	\gnArgHgt=\ht\gnBoxA%
	\ifnum \gnArgHgt<\gnCornerHgt
		\gnArgHgt=0pt%
	\else
		\advance \gnArgHgt by -\gnCornerHgt%
	\fi
	\raise\gnArgHgt\hbox{$\ulcorner$} \box\gnBoxA %
		\raise\gnArgHgt\hbox{$\urcorner$}}

\def\bdefi{\begin{defi}\rm}
\def\edefi{\end{defi}}
\def\bnota{\begin{nota}\rm}
\def\enota{\end{nota}}
\def\brem{\begin{rem}\rm}
\def\erem{\end{rem}}

\def\NSA{\mathbb{NSA}}

\def\RCA{\textup{RCA}}
\def\LPRO{\mathbb{LPR}}

\def\ASA{\Lleftarrow\!\!\Rrightarrow}

\def\WKL{\textup{WKL}}

\def\bye{\end{document}}

\def\N{{\mathbb  N}}
\def\Q{{\mathbb  Q}}
\def\R{{\mathbb  R}}

\def\DI{\Rrightarrow}

\def\WLPO{{\textup{WLPO}}}
\def\LPO{{\textup{LPO}}}
\def\LLPO{{\textup{LLPO}}}
\def\MP{{\textup{MP}}}

\def\MPO{{\textup{MP}^{\vee}}}
\def\LPOO{{\mathbb{LPO}}}
\def\LLPOO{{\mathbb{LLPO}}}
\def\MPO{{\mathbb{MP}}}

\def\WLPOO{{\mathbb{WLPO}}}

\def\MCTO{\mathbb{MCT}}

\def\sN{{^{*}\mathbb{N}}}

\def\T{{\mathbb  T}}
\def\R{{\mathbb{R}}}
\def\({\textup{(}}
\def\){\textup{)}}

\def\asa{\leftrightarrow}

\def\di{\rightarrow}

\def\eps{\varepsilon}

\def\vx{\vec{x}}

\def\paai{\Pi_{1}\textup{-TRANS}}

\def\HD{\mathop{\mathbb{V}}}

\def\titlerunning{A new model of computation in nonstandard analysis}
\def\authorrunning{S.\ Sanders}

\title{Algorithm and proof as $\Omega$-invariance and transfer: A new model of computation in nonstandard analysis}

\author{Sam Sanders\thanks{This publication was made possible through the generous support of a grant from the John Templeton Foundation for the project \emph{Philosophical Frontiers in Reverse Mathematics}. I thank the John Templeton Foundation for its continuing support for the Big Questions in science.  Please note that the opinions expressed in this publication are those of the author and do not necessarily reflect the views of the John Templeton Foundation.}
\institute{Department of Mathematics, Universiteit Gent, Krijgslaan 281, 9000 Gent, Belgium}
\email{sasander@cage.ugent.be}}

\begin{document}
\maketitle

\begin{abstract}
We propose a new model of computation based on Nonstandard Analysis.
Intuitively, the role of `algorithm' is played by a new notion of finite procedure, called $\Omega$-invariance and inspired by physics, from Nonstandard Analysis.
Moreover, the role of `proof' is taken up by the Transfer Principle from Nonstandard Analysis.  We obtain a number of results in Constructive Reverse Mathematics to
illustrate the tight correspondence to Errett Bishop's \emph{Constructive Analysis} and the associated \emph{Constructive Reverse Mathematics}.
\end{abstract}


\section{Introduction}\label{intro}
Historically, the first models of computation go back to Alan Turing and Alonzo Church \cite{tochurchonsunday,turing37}.
In order to solve David Hilbert's famous \emph{Entscheidungsproblem} in the negative, Church and Turing both proposed a formal definition of algorithm (resp.\ the Turing machine and the $\lambda$-calculus).
In time, it became clear that these formalisms describe the same computational phenomenon, i.e., the recursive functions.
The importance of the discovery of the Turing machine and the $\lambda$-calculus for (what later would become) Computer Science cannot be underestimated.  Nonetheless, these formalisms have certain limitations.  Foremost, they operate on \emph{discrete} data, i.e., natural numbers (in some representation or other) are the basic concept.
The question naturally arises whether there is a theory of computability and complexity over \emph{real-valued data}, i.e., where the basic concept is that of real number.
An example of such a formalism is \emph{Computable Analysis} (See, e.g., \cite{brats,weikeei}).  Moreover, there is even a conference series\footnote{See the website \url{http://cca-net.de/events.html}.}, called {CCA}, dedicated to computability over the reals.  The following is an excerpt from the scope of CCA.  
\begin{quote}
\emph{The conference is concerned with the theory of computability and complexity over real-valued data.}
[\dots]
\emph{Most mathematical models in physics and engineering \emph{[\dots]} are based on the real number concept. Thus, a computability theory and a complexity theory over the real numbers and over more general continuous data structures is needed.}
\end{quote}
%
An American wordsmith might at this point refer to the legal oath concerning truthful testimony.  Indeed, while the observations made in the above quote are certainly valid, and the existing frameworks for real-number computability constitute highly respectable academic disciplines, the above observation is not \emph{the whole truth}.
For instance, an equally valid observation concerning physics and engineering is that an intuitive \emph{calculus with infinitesimals}, i.e., informal Nonstandard Analysis, is freely employed.
In particular, the notorious `$\eps$-$\delta$ method', due to Karl Weierstra\ss, was never adopted, although it hitherto constitutes the \emph{de facto standard} in mathematics, especially analysis.
These observations concerning physics and engineering beg the question whether
a notion of computability can be developed which is directly based
on infinitesimals.  In this paper, we sketch what such a formalism might look like.
In particular, we define a notion of \emph{finite procedure}, called $\Omega$-invariance, inside Nonstandard Analysis.
We also define three extra connectives $\HD$, $\DI$, and ${\sim}$, called \emph{hyperconnectives} inside (classical) Nonstandard Analysis.

\medskip

We arrive at the definition of the hyperconnectives motivated
by observations regarding infinitesimals in the mathematical practice of physics and engineering listed after Example \ref{firstexa}.
As it turns out, our new notion of algorithm and our new connectives bear a remarkable resemblance to the constructive
versions of algorithm and the connectives 
in BISH, i,e,\ Errett Bishop's \emph{Constructive Analysis} \cite{bish1,bridge1}.
Indeed, the definitions of $\DI$, $\HD$ and ${\sim}$ follow the well-known BHK (Brouwer-Heyting-Kolmorgorov) interpretation of intuitionistic logic \cite{buss}, but with provability and algorithm replaced by Transfer and $\Omega$-invariance, respectively.
In particular, we shall show that the hyperconnectives give rise to the same equivalences of Constructive Reverse Mathematics \cite{ishi1} \emph{inside} Nonstandard Analysis.
The following example is representative of our results.
\begin{exa}\label{firstexa}\rm
In BISH, LPO is the statement \emph{For all $P\in \Sigma_{1}$, we have $P\vee \neg P$}, i.e., the law of excluded middle limited to $\Sigma_{1}$-formulas.
Also, LPR is the statement that $(\forall x\in \R)(x>0 \vee \neg(x>0))$, the statement that the continuum can be split in two parts.
Now, BISH proves that LPO and LPR are equivalent, and neither principle is provable in BISH, for obvious reasons.
By Theorem \ref{uni2}, the equivalence between $\LPRO$, i.e., $(\forall x\in \R)(x>0 \HD {\sim}(x>0))$ and $\LPOO$, i.e., the statement \emph{For all $P\in \Sigma_{1}$, we have $P\HD {\sim}P$} is provable in a suitably weak system of Nonstandard Analysis, called $\NSA$, which can prove neither principle outright.
 \end{exa}
Let us first consider two observations regarding the use of infinitesimals in physics.  We assume the reader enjoys basic familiarity with Nonstandard Analysis, pioneered by Abraham Robinson \cite{robinson1}.

\medskip

First of all, we argue that end results of calculations (involving infinitesimals) in physics do not depend on the \emph{choice} of the infinitesimal used.
Indeed, in the course of providing a version of each of the \emph{Big Five} systems of Reverse Mathematics \cite{simpson2,keisler1},  Keisler has shown that sets of natural numbers (or, equivalently, real numbers) can be coded by infinite numbers, even in relatively weak systems\footnote{In particular, Keisler has shown that the statement \emph{Every set of natural numbers $X$ is coded by some infinite number $x$}, can be conservatively added to a nonstandard extension of $\RCA_{0}$.  Furthermore, the addition of the statement \emph{Every infinite number is the code for some set of natural numbers $X$}, results in a nonstandard extension of $\WKL_{0}$.   }.
Hence, a particular infinite number $\omega$ can contain a lot of information, as it may, e.g., code the Halting Problem.  In this way, a formula $\varphi(\omega)$, which depends on $\omega$,
may not be decidable, even if $\varphi$ has low complexity, e.g., $\Delta_{0}$.  However, for another infinite $\omega'$, say which codes the set of prime numbers, and the same $\Delta_{0}$-formula $\varphi$,
the formula $\varphi(\omega')$ would be intuitively decidable.  This key insight gives rise to our notion of finite procedure in Nonstandard Analysis in Definition \ref{kref}.

\medskip

Indeed, infinitesimals (which are the inverses of infinite numbers) are used freely in physics and engineering, in an informal way.
Nonetheless, they are generally believed to be mere calculus tools, whose sole purpose is to simplify mathematical derivation and calculation.
Furthermore, as physics and engineering deal with the modeling of the physical world,
the end result of a calculation (even using `ideal' objects such as infinitesimals) should have \emph{real-world meaning}.  Thus, the physically meaningful end result should not
depend on the \emph{choice} of infinitesimal used in the calculation, as common sense dictates that the real world does not depend on our arbitrary choice of calculus tool (i.e., infinitesimal).
Simply put, if we repeat the same calculation with a different infinitesimal, we expect that we should obtain the same result.
Hence, our first observation is that the end results in physics do not depend on the \emph{choice} of infinitesimal used in their derivation.  This is formalized directly in the notion of $\Omega$-invariance: our notion of finite procedure in Nonstandard Analysis, introduced in Definition \ref{kref} below.

\medskip

Secondly, we consider a universal formula of the form $(\forall n\in \N)\varphi(n)$, with $\varphi$ in $\Delta_{0}$.
By G\"odel's famous incompleteness theorem (See, e.g., \cite{buss}), we know that, no matter how obvious, appealing, true, natural, or
convinving $(\forall n\in \N)\varphi(n)$ may seem, this formula may not be provable in the formal system at hand.
Indeed, the consistency of a given system $T$ is a universal formula, and is in general not provable in $T$.  Hence, there might be (nonstandard) models in which there is a (nonstandard) $n$ such that $\neg\varphi(n)$.
While G\"odel's `unprovability' results are extremely interesting from the point of view of logic and mathematics, it is a phenomenon one would like to avoid in physics.
Indeed, it seems intuitively appealing that `ideal' nonstandard elements (e.g., infinitesimals and infinite numbers) satisfy the same properties as the `basic' standard objects.
This observation essentially dates back to Leibniz \cite{evenbellen} and is often called \emph{Leibniz' law}.  It is formalized in the \emph{Transfer Principle} of Nonstandard Analysis, but it turns out that even weak versions of this principle have too great a logical strength for systems that pertain to physical reality \cite{aloneatlast3}.
Furthermore, Theorem \ref{firstlast2} and \cite[Proposition 2.1]{palmdijk} also suggest that the Transfer Principle should be omitted.

\medskip

Besides the Transfer Principle, another obvious way of circumventing the existence of nonstandard counterexamples is to just demand (in the nonstandard model $\sN$ of $\N$ at hand) that $(\forall n\in \N)\varphi(n)$ satisfies
\be\label{frik}
(\forall n\in \N)\varphi(n) \wedge [(\forall n\in \N)\varphi(n) \di (\forall n\in \sN)\varphi(n)],
\ee
and to ignore $(\forall n\in \N)\varphi(n)$ otherwise.  The previous constitutes our second observation, to be compared to the definition of $\T$ below.
In a sense, this is also the avenue taken in Constructive Analysis: unprovable statements are essentially ignored by the very interpretation of the logical connectives.  
\section{A notion of finite procedure in Nonstandard Analysis}\label{owmega}
In this section, we introduce $\Omega$-invariance, the formalization of the intuitive idea from Section \ref{intro} that end results in physics do not depend on the \emph{choice} of infinitesimals used.
We show that $\Omega$-invariance is quite close to the notion of finite procedure.

\medskip

With regard to notation, we take $\N=\{0,1,2,\dots\}$ to denote the set of natural numbers,
which is extended to $\sN=\{0,1,2,\dots, \omega, \omega+1,\dots\}$, the set of \emph{hypernatural} numbers, with $\omega\not\in\N$.
The set $\Omega=\sN\setminus \N$ consists of the \emph{infinite} numbers, whereas the natural numbers are \emph{finite}.
Finally, a formula is \emph{bounded} or `$\Delta_{0}$', if all the quantifiers are {bounded} by terms and no infinite numbers occur.
\begin{defi}[$\Omega$-invariance]\label{kref} Let $\psi(n,m)$ be $\Delta_{0}$ and fix $\omega\in \Omega$.  \\
The formula $\psi(n,\omega)$ is \emph{$\Omega$-invariant} if  
\be\label{omegainv}
(\forall n\in \N)(\forall \omega'\in \Omega)(\psi(n,\omega)\asa \psi(n,\omega')).
\ee
For $f:\N\times\N\di \N$, the function $f(n,\omega)$ is called $\Omega$-invariant, if  
\be
(\forall n\in \N)(\forall \omega'\in \Omega)(f(n,\omega)=f(n,\omega')).
\ee
\edefi
As mentioned above, any object $\varphi(\omega)$ defined using an infinite number $\omega$ is potentially non-computable or undecidable, as infinite numbers can code (non-recursive) sets of natural numbers.
Hence, it is not clear how an $\Omega$-invariant object might be computable or constructive in any sense.
However, although an $\Omega$-invariant object clearly involves an infinite number, the object does not depend on the \emph{choice} of the infinite number, by definition.
Furthermore, by the following theorem, the truth value of $\psi(n,\omega)$ and the value of $f(n,\omega)$ is \emph{already} determined at some finite number.  
\begin{thm}[Modulus lemma]\label{moduluslemma}
For every $\Omega$-invariant formula $\psi(n,\omega)$,
\[
(\forall n\in \N)(\exists m_{0}\in \N)(\forall m,m'\in \sN)[m,m'\geq m_{0}\di \psi(n,m)\asa \psi(n,m')].
\]
For every $\Omega$-invariant function $f(n,\omega)$, we have
\[ 
(\forall n\in \N)(\exists m_{0}\in \N)(\forall m,m'\in \sN)[m,m'\geq m_{0}\di f(n,m)= f(n,m')].
\]
In each case, the number $m_{0}$ is computed by an $\Omega$-invariant function.
\end{thm}
\begin{proof}
Follows easily using underspill for $\Delta_{0}$-formulas.
\end{proof}
The previous theorem is called `modulus lemma' as it bears a resemblance to the modulus lemma from Recursion Theory \cite[Lemma 3.2]{zweer}
Intuitively, our modulus lemma states that the properties of an $\Omega$-invariant object are already determined at some \emph{finite} number.
This observation suggests that the notion of $\Omega$-invariance models the notion of \emph{finite procedure}.

\medskip

Another way of interpreting $\Omega$-invariance is as follows: Central to any version of constructivism is that there are basic objects (e.g., the natural numbers) and there are certain basic operations on
these objects (e.g., recursive functions or constructive algorithms).  All other objects are non-basic (aka `non-constructive' or `ideal'), and are to be avoided, as they fall outside the constructive world.
It goes without saying that infinite numbers in $\sN$ are ideals objects \emph{par excellence}.
Nonetheless, our modulus lemma suggests that if an object does not depend on the \emph{choice} of ideal element in its definition, it is not ideal, but actually basic.
This is the idea behind $\Omega$-invariance: ideal objects can be basic if their definition does not really depend on the choice of any particular ideal element.  In this way, $\Omega$-invariance approaches
the notion of finite procedure \emph{from above}, while the usual methods work \emph{from the ground up} by defining a set of basic constructive operations and a method for combining/iterating these.

\medskip

Finally, note that if $(\forall n\in\N)\varphi(n)$ satisfies \eqref{frik}, then the formula $\psi(\omega)\equiv(\forall n\leq \omega)\varphi(n)$ is $\Omega$-invariant.
Moreover, we have
\be\label{cccc}
\psi(\omega)\di (\forall n\in\N)\varphi(n) \wedge \neg\psi(\omega) \di (\exists n\in\N)\neg\varphi(n).
\ee
Intuitively, the previous formula suggests that there is an $\Omega$-invariant decision procedure for $(\forall n\in\N)\varphi(n)$.
Note that $(\forall n\in\N)\varphi(n)$ can be replaced by $(\forall n\in\sN)\varphi(n)$ in \eqref{cccc}.  This naturally leads to the definition of $\T$ below in Defintion \ref{bigT}.
\section{The base system $\NSA$}\label{owgiga}
In this section, we introduce the base theory $\NSA$ in which we shall work.  The system $\NSA$ is inspired by the theories ns-$\WKL_{0}$, $^{*}\RCA_{0}$, $^{*}\RCA_{0}'$, and $^{*}\WKL_{0}$, introduced by Keita Yokoyama and Jerry Keisler \cite{kei1,kei2,kei3,kei4, keisler1,keisler2} to apply techniques of Nonstandard Analysis in Reverse Mathematics.  Like Yokoyama's system ns-$\WKL_{0}$, we consider more than two sorts, corresponding to the standard number variables,
the nonstandard number variables, the standard set variables, and the nonstandard set variables.  The standard number and set variables are intended to range over the natural numbers, and subsets thereof.
The nonstandard number and set variables are intended to range over a nonstandard extension of the natural numbers, and subsets thereof.
Finally, there is an embedding from the standard structure into the nonstandard structure, and the usual `star morphism' $^{*}$ from Nonstandard analysis behaves as expected.  

\medskip

We shall denote the set of standard (or finite) numbers by $\N$, and the set of nonstandard numbers by $\sN$.  The set $\Omega$ is $\sN\setminus \N$, the set of infinite numbers.
For brevity, we use more informal notation than Yokoyama and Keisler.
\bdefi[s-$\NSA$]\label{defP}~
\begin{enumerate}
\item (Induction) The axioms $I\Sigma_{1}$ and $^{*}I\Sigma_{1}$.
\item (End Extension) The set $\sN$ is an end extension of $\N$.
\item {(Embedding)} The standard structure is embedded in the nonstandard structure via an injective homomorphism.
\item\label{oca} ($\Omega$-CA): For all $\psi(n,m)\in \Delta_{0}$ and $\omega\in\Omega$,
\[
(\forall n\in\N)(\forall \omega' \in \Omega)[^{*}\psi(n,\omega)\asa {^{*}\psi}(n,\omega')]\di (\exists X\subset \N)(\forall n\in\N)[^{*}\psi(n,\omega)\asa n\in X].
\]
\item\label{trakje} ($\Delta_{0}$-TRANS): $(\forall n\in \N)(\forall X\subset \N)[\varphi(n,X) \asa {^{*}\varphi(n,X)}]$, for $\varphi\in \Delta_{0}$.
\item ($\Delta_{0}$-underspill): $(\forall n\in \N)(\forall X\subset \sN)[(\forall \omega\in \Omega){^{*}\varphi}(n,X,\omega)\di (\exists m\in \N)\varphi(n,X,m)]$.  \label{bukaka}
\end{enumerate}
\edefi
Here, item \ref{oca} expresses that certain nonstandard objects (the $\Omega$-invariant formulas) may be used to define standard sets.
Similar principles are usually called \emph{standard part principles} \cite{keisler1}.
Also, item \ref{trakje} guarantees that functions defined on $\N$ (like $2^{n}$), are also defined on $\sN$.  We shall sometimes be sloppy concerning the use of the star morphism `$^{*}$'.
For instance, we just write $2^{\omega}$ instead of $^{*}2^{\omega}$, if $\omega\in\Omega$.

\medskip

In light of our modulus lemma, the sets definable via comprehension in s-$\NSA$ are those for which there is a finite procedure to decide elementhood.
In practice, item \ref{bukaka} (underspill) in Definition \ref{defP} is never used and may thus be omitted to ensure that not all $\Delta_{1}$-formulas are decidable.
Finally, we assume that s-$\NSA$ satisfies the following \emph{transfer rule}: If s-$\NSA$ proves $(\forall n \in\N)\varphi(n)$, for $\varphi \in \Delta_{0}$, then s-$\NSA$ proves $(\forall n\in\sN){^{*}}\varphi(n)$.

\medskip

As it turns out, s-$\NSA$ does not suffice for our purposes.  In particular, the binary distinction between finite and infinite numbers is `too coarse'.  We define $\NSA$ in the same way as s-$\NSA$,
but with at least one explicit extension $\N_{1}$ such that $\N\subsetneq \N_{1} \subsetneq \sN$.
The new extension $\N_{1}$ is connected to the other levels by $\Delta_{0}$-transfer, in the same way as $\N$ and $\sN$ are, and underspill also holds between levels.
The elements of $\N_{1}$ are called `$1$-finite' and the set $\Omega_{1}=\sN \setminus \N_{1}$ consists of the `$1$-infinite' numbers.  Intuitively, the set $\sN$ is \emph{stratified} into two levels of infinity.
Indeed, the elements of $\N_{1}\setminus \N$ are \emph{infinite}, but they are \emph{$1$-finite}, and the elements of $\Omega_{1}$ are infinite \emph{and $1$-infinite}.
The theory $\NSA$ is an example of \emph{stratified} or \emph{relative} Nonstandard Analysis \cite{peraire1,aveirohrbacek,hrbacek2,aloneatlast1,aloneatlast2}.
In case there are countably many $\N_{k}$ such that
\[
\N=\N_{0}\subsetneq \N_{1}\subsetneq \N_{2}\subsetneq \dots \N_{k}\subsetneq \N_{k+1} \subsetneq\dots \subsetneq {\sN},
\]
then for $k=0$, the index `$0$' is usually omitted, i.e., $\N=\N_{0}$ consists of the finite numbers and $\Omega=\sN \setminus \N$ are the infinite numbers.

\medskip

Finally, we consider the following Transfer Principle, not available in $\NSA$, but needed below.
\begin{princ}[$\paai$]
For all $\varphi$ in $\Delta_{0}$, we have
\be\label{xyz}
(\forall n\in\N)\varphi(n) \di (\forall n\in \sN)^{*}\varphi(n).
\ee
\end{princ}
We always assume that $\Delta_{0}$-formulas may contain standard sets and numbers as parameters.
\section{Reverse-engineering Constructive Reverse Mathematics}
\emph{Constructive Reverse  Mathematics} is a spin-off from Harvey Friedman's famous \emph{Reverse Mathematics} program \cite{fried,fried2} where the base theory is (based on) BISH.
In this section, we define the new hyperconnectives and show that many of the well-known principles from Constructive Reverse Mathematics (e.g., LPO and MP) satisfy
the same equivalences as their counterparts inside $\NSA$ (e.g., $\LPOO$ and $\MPO$).  The latter are obtained from the former by simply replacing the usual connectives
by the new hyperconnectives as in Example \ref{firstexa}.
%
%
\bdefi[Hyperconnectives]\label{hypercon} For formulas $A$ and $B$,
\begin{enumerate}
\item $A\DI B$ is defined as $[A \wedge A\in \T] \di[ B \wedge B\in \T]$, 
\item ${\sim}A$ is defined as $A\DI 0=1$,
\item $A\HD B$ is defined as \emph{There is an $\Omega$-invariant formula $\psi(\vx,\omega)$ such that}
\be\label{jokkk}
(\forall \vx\in \N^{k})\big(\psi(\vx,\omega)\di [A(\vx)\wedge A(\vx)\in \T] \wedge \neg\psi(\vx,\omega)\di [B(\vx)\wedge B(\vx)\in \T]\big),
\ee
for formulas $A,B$ involving parameters $\vx\in \N^{k}$.
\end{enumerate}
\edefi
The third connective is called `hyperdisjunction', the second and third one are called `hyperimplication' and `hypernegation'.
Obviously, this definition hinges on the definition of $\T$.  Intuitively, for suitable $A$, the formula `$A\in \T$' expresses that $A$ satisfies the transfer principle of Nonstandard Analysis.
For brevity, we give the definition of $\T$ up to $\Sigma_{2}\cup\Pi_{2}$-formulas.  The definition for higher complexity is analogous, straightforward, and clear from the following definition.
\bdefi[The set $\T$]\label{bigT} For $\varphi\in\Delta_{0}$, we define
\begin{enumerate}
\item The formula $(\forall n\in \N)\varphi(n) \in \T$ is $(\forall n\in \N)\varphi(n) \di (\forall n\in \sN)\varphi(n)$.
\item The formula $(\exists n\in \sN)\varphi(n) \in \T$ is $(\exists n\in \sN)\varphi(n) \di (\exists n\in \N_{1})\varphi(n)$.
\item The formula $(\forall n\in \N)(\exists m\in \N)\varphi(n,m) \in \T$ is
\[
(\forall n\in \N)(\exists m\in\N)\varphi(n,m) \di (\forall n\in \sN)(\exists m\in \sN)[\text{$m$ is $n$-finite} \wedge \varphi(n,m)].
\]
\item The formula $(\exists n\in \sN)(\forall m\in\sN)\varphi(n,m) \in \T$ is
\[
(\exists n\in \sN)(\forall m\in \sN)\varphi(n,m) \di (\exists n\in \N_{1})(\forall m\in\sN)\varphi(n,m).
\]
\item For formulas $A$ not suitable for transfer, like $(\exists n\in \N)\varphi(n)$, $(\exists n \in \N_{1})\varphi(n)$, $0=1$, and $(\forall n\in \sN)\varphi(n)$, the formula $A\in \T$ is just $0=0$.  In this case, $A\in \T$ is called `trivial'.
\end{enumerate}
\edefi
Note that universal formulas in $\T$ `maximally' satisfy the transfer principle\footnote{E.g.\ we have $(\forall n\in \N)\varphi(n) \di (\forall n\in \sN)\varphi(n)$, and not $(\forall n\in \N)\varphi(n) \di (\forall n\in \N_{1})\varphi(n)$.}, but existential formulas only satisfy the transfer principle \emph{maximally minus one level}.
With the above definitions, it is clear that $\DI$, $\HD$ and ${\sim}$ follow the well-known BHK interpretation of intuitionistic logic \cite{buss}, but with provability and algorithm replaced by Transfer and $\Omega$-invariance, respectively.
\begin{rem}\rm
An alternative definition of hyperdisjunction is the following:  
\[
A\mathcal{V}B \equiv\big( [A\vee B] \wedge[ A\di A\in \T] \wedge [B\di B\in \T]\big).
\]
For $\mathcal{V}$ instead of $\HD$, we would obtain results similar to the ones below, except that $\LPOO$, $\LLPOO$, $\WLPOO$, and $\MPO$ would all be equivalent to $\paai$.
Intuitively, this corresponds to the extreme intuitionistic view that $\LPO$, $\LLPO$, $\WLPO$ and $\MP$ are all equally `false' principles.
\end{rem}
\begin{rem}\rm The definition of $\HD$ is based on the observation that $\paai$ is a kind of `hyperexcluded middle': it excludes the possibility that $(\forall n\in \N)\varphi(n)\wedge (\exists n\in\sN)\neg\varphi(n)$.
This informal observation is formalized in Example \ref{firstexa} and Theorem \ref{firstlast2}.  Moreover, given $\paai$, $(\forall n\in \N)\varphi(n)$ is equivalent to
$(\forall n\leq \omega)\varphi(n)$, for any choice of $\omega\in\Omega$, as noted in \cite{tale,aloneatlast3}.  Hence, an $\Omega$-invariant formula can decide if $(\forall n\in \N)\varphi(n)$ is true or not, suggesting \eqref{jokkk}.  Furthermore, under certain conditions, if a system of Nonstandard Analysis $T^{ns}$
proves a standard universal formula $(\forall n\in \N)\varphi(n)$, then $T^{ns}$ also proves $(\forall n\in \sN)\varphi(n)$, i.e., we have $(\forall n\in \N)\varphi(n)\in\T$.
In light of the BHK-interpretation of implication, the previous observation motivates the definition of $\T$ and $\DI$.
\end{rem}
\subsection{Limited Omniscience}
In this paragraph, we consider LPO, the \emph{Limited Principle of Omniscience}, which is the Law of Excluded Middle (LEM), limited to $\Sigma_{1}$-formulas.
This principle is considered to be the original sin of classical mathematics according to constructivist canon \cite{vajuju, bish1}.
\begin{princ}[$\LPOO$]
For all $P\in \Sigma_{1}$, we have $P\HD {\sim}P$.
\end{princ}
\begin{thm}\label{firstlast2}
In $\NSA$, $\LPOO$ is equivalent to $\paai$.
\end{thm}
\begin{proof}
Assume $\LPOO$ and let $\varphi(n,\vx)$ be as in $\paai$ with all parameters $\vx\in \N^{k}$ shown.
Then we have
\be\label{kerkje2}
(\exists n\in \N)\varphi(n,\vx) \HD {\sim}[(\exists n\in \N)\varphi(n,\vx)].
\ee
By definition, ${\sim}[(\exists n\in \N)\varphi(n,\vx)]$ is $(\exists n\in \N)\varphi(n,\vx)\DI 0=1$, which is
\be\label{ella}
\big[(\exists n\in \N)\varphi(n,\vx) \wedge[ (\exists n\in \N)\varphi(n,\vx)\in \T] \big]\di \big[0=1 \wedge[ 0=1\in \T]\big],
\ee
again by definition.
As neither $(\exists n\in \N)\varphi(n,\vx)$ nor $0=1$  is suitable for transfer, \eqref{ella} reduces to $(\exists n\in \N)\varphi(n,\vx)\di 0=1$, i.e., $(\forall n\in\N)\neg\varphi(n,\vx)$.
Indeed, by definition, the formulas $(\exists n\in \N)\varphi(n,\vx)\in \T$ and $0=1\in \T$ are $0=0$, i.e., trivial.   
Again by definition, \eqref{kerkje2} implies the existence of some $\Omega$-invariant $\psi(\vx,\omega)$ such that, for all $\vx\in \N^{k}$,
\[
\psi(\vx,\omega)\di (\exists n\in \N)\varphi(n,\vx)
\]
\be\label{mini2}
\wedge
\ee
\[
  \neg\psi(\vx,\omega) \di \big[(\forall n\in \N)\neg\varphi(n,\vx)  
\wedge
[ (\forall n\in \N)\neg\varphi(n,\vx)  \di (\forall n\in \sN)\neg\varphi(n,\vx)] \big].
\]
Note that the innermost square-bracketed subformula is $(\forall n\in \N)\neg\varphi(n)\in \T$.
Now suppose $(\exists n\in \sN)\varphi(n,\vx_{0})$ for some fixed $\vx_{0}\in\N^{k}$, and consider \eqref{mini2}.  Then $\neg\psi(\vx_{0},\omega)$ would imply $(\forall n\in \sN)\neg\varphi(n,\vx_{0})$, which is impossible.  However, $\psi(\vx_{0},\omega)$ implies $(\exists n\in \N)\varphi(n,\vx_{0})$, and we have derived $(\exists n\in \sN)\varphi(n,\vx_{0})\di (\exists n\in \N)\varphi(n,\vx_{0})$,
an instance of $\paai$, using \eqref{mini2}, i.e., $\LPOO$.  Hence $\LPOO$ implies $\paai$ and this case is done.

\medskip

For the reverse direction, assume $\paai$ and let $\varphi(n,\vx)$ be as in $\LPOO$.  By $\paai$, $(\exists n\in \N)\varphi(n,\vx)$ is equivalent to $(\exists n\leq \omega)\varphi(n,\vx)$, for any $\omega\in \Omega$ and $\vx\in \N^{k}$.  In other words, the formula $\psi(\vx,\omega)\equiv (\exists n\leq \omega)\varphi(n,\vx)$ is $\Omega$-invariant and satisfies \eqref{mini2}.  As the latter is equivalent to \eqref{kerkje2}, an instance of $\LPOO$, we obtain the latter principle.
\end{proof}
To consider principles regarding real numbers, we first need to define these objects.  We use the usual definition involving fast-converging sequences from \cite{simpson2}.
\begin{defi}[Real number]\label{keepinitreal2} \rm
A real number is a sequence of rationals $(q_{n})$ s.t.\ $(\forall k,i\in\N)(|q_{k}-q_{k+i}|< \frac{1}{2^{k}})$.
Two reals $x=(q_{n})$ and $y=(r_{n})$ are equal, i.e., $x=y$, if $(\forall k\in \N)(|q_{k}-r_{k}|\leq \frac{1}{2^{k-1}})$, $x<y$ is $(\exists k\in \N)(q_{k}+\frac{1}{2^{k}}<r_{k})$, and $x\geq y$ is $\neg(x<y)$.
\end{defi}
\bnota
The reals $x$ and $y$ are always assumed to be given as $x=(q_{n})$ and $y=(r_{n})$, for sequences $q_{n}, r_{n}:\N \di \Q$.
\enota
We now consider the following principle.
\begin{princ}[$\mathbb{LPR}$]
$(\forall x\in \R)(x>0 \HD {\sim}(x>0))$.
\end{princ}
We have the following theorem.
\begin{thm}\label{uni2}
In $\NSA$, $\LPOO$ is equivalent to $\mathbb{LPR}$.
\end{thm}
\begin{proof}
For the forward direction, we may assume $\paai$, by Theorem \ref{firstlast2}.
By definition, $x>0$ is $(\exists n\in \N)(q_{n}>\frac{1}{2^{n}})$, and ${\sim}{(x>0)}$ is $(\forall n\in \N)(q_{n}\leq\frac{1}{2^{n}})$.
In the same way as in the proof of Theorem \ref{firstlast2}, the formula $\psi(\omega)\equiv(\exists n\leq \omega)(q_{n}>\frac{1}{2^{n}})$ is equivalent to $x>0$, by $\paai$.
Note that $q_{n}$, and hence $\psi$, may depend on standard parameters $\vx\in\N^{k}$ which are not shown.  By the above, we have
\[
\psi(\omega)\di (x>0) \wedge \neg\psi(\omega)\di {\sim}(x>0).
\]
By definition, $x>0 \in \T$ is trivial and, by $\paai$, we have ${\sim}(x>0)\in \T$, as the latter is just $(\forall n\in \N)(q_{n}\leq\frac{1}{2^{n}})\di (\forall n\in \sN)(q_{n}\leq\frac{1}{2^{n}})$.
Hence, $x>0 \HD {\sim}(x>0)$ follows immediately.  Finally, for the formula $(\forall x)(x\in \R \DI x>0 \HD {\sim}(x>0))$, simply apply $\Pi_{1}$-transfer to `$x\in \R$', which is $\Pi_{1}$, to
observe that the latter is in $\T$.

\medskip

For the reverse direction, assume $\LPRO$, and let $\varphi$ be as in $\paai$.
Now define\footnote{For $\varphi\in\Delta_{0}$, there is a function $T_{\varphi}$ which is $0$ or $1$ depending on whether $\varphi$ is false or true.} the real $x$ as $q_{n}:=\sum_{i=0}^{n}\frac{1}{2^{i}}T_{\psi}(i)$, for $\psi(i)$ defined as $(\exists m\leq i)\varphi(m)$.
Note that $x$ indeed satisfies the definition of real number from Definition \ref{keepinitreal2}, as we have
\[\textstyle
(\forall n,m\in\N)\big(m<n \di |q_{m}-q_{n}|\leq \sum_{i=m+1}^{n}\frac{1}{2^{i}} <\frac{1}{2^{m}} \big).
\]
Note that the previous formula is also valid for $\sN$ instead of $\N$.  In other words, `$x\in\R$' is in $\T$.

\medskip

By $\LPRO$, we either have $x>0$, i.e., $(\exists n\in \N)(q_{n}>\frac{1}{2^{n}})$, or ${\sim}(x>0)$, i.e., $ (\forall n\in \N)(q_{n}\leq \frac{1}{2^{n}})$, \emph{and} the fact that ${\sim}(x>0)$ is in $\T$.
The latter is
\be\label{friekske2}\textstyle
(\forall n\in \N)(q_{n}\leq \frac{1}{2^{n}})\di (\forall n\in \sN)(q_{n}\leq \frac{1}{2^{n}}).
\ee
Thus, if $(\forall n\in \N)\neg\varphi(n)$, then $q_{n}=0$, for $n\in \N$, and also $(\forall n\in \N)(q_{n}\leq \frac{1}{2^{n}})$.
By \eqref{friekske2}, we have $(\forall n\in \sN)(q_{n}\leq \frac{1}{2^{n}})$, which implies that $(\forall n\in \sN)\neg\varphi(n)$.  Indeed, if $(\exists n\in \sN)\varphi(n)$, say we have $\varphi(n_{0})$, then
$q_{n}>q_{n_{0}}\geq \frac{1}{2^{n_{0}}}$, for all $n>n_{0}$.  Hence, we have $q_{n_{0}+1}>\frac{1}{2^{n_{0}+1}}$, which would contradict $(\forall n\in \sN)(q_{n}\leq \frac{1}{2^{n}})$.

\medskip

Thus, given $\LPRO$, the formula $(\forall n\in \N)\neg\varphi(n)$ implies $(\forall n\in \sN)\neg\varphi(n)$, which is an instance of $\Pi_{1}$-TRANS.  By Theorem \ref{firstlast2}, $\LPOO$ follows and we are done.
\end{proof}
\begin{cor}
In $\NSA$, we have both $\LPOO\asa \LPRO$ and $\LPOO\ASA \LPRO$.
\end{cor}
\begin{proof}
The first equivalence was proved in the theorem.  For the second (hyper)equivalence, we note that $\LPOO \ASA \LPRO$ reduces to $\LPOO\asa \LPRO$, as both $\LPOO$ and $\LPRO$ are not (equivalent to) standard formulas.
Indeed, otherwise $\NSA$ would include $\paai$, which is not the case.
\end{proof}
We assume the usual (classical) definitions of sequence of reals, Cauchy sequence, convergent and increasing sequence.
\begin{princ}[$\MCTO$]
If $x_{n}$ is a bounded and increasing sequence of reals, then $x_{n}$ converges to a limit:
\[\textstyle
(\forall n\in \N)(x_{n}\leq x_{n+1})\wedge (\exists N\in \N)(\forall n\in \N)( x_{n} \leq N )
\]
\vspace{-0,7cm}
\be\label{BIG}
\rotatebox[origin=c]{270}{$\DI$}
\ee
\[\textstyle
(\exists x \in \R)(\forall k\in \N)(\exists M\in \N)(\forall m\in\N)(m\geq M\DI |x_{m}-x|\leq \frac{1}{k})
\]
\end{princ}

\begin{thm}
In $\NSA$, $\mathbb{MCT}$ is equivalent to $\LPOO$.
\end{thm}
\begin{proof}
We first prove the forward implication.  It is not difficult to prove that $\MCTO$ is equivalent to
\[\textstyle
(\forall n\in \N)(x_{n}\leq x_{n+1})\wedge (\exists N\in \N)(\forall n\in \N)( x_{n} \leq N )
\]
\vspace{-0.7cm}
\be\label{BIGGER2}
\rotatebox[origin=c]{270}{$\DI$}
\ee
\[\textstyle
(\forall k\in \N)(\exists M\in \N)(\forall m,m'\in\N)(m,m'\geq M\DI |x_{m}-x_{m'}|\leq \frac{1}{k}),
\]
as a sequence is convergent if and only if it is a Cauchy sequence\footnote{Incidentally, this is a theorem of BISH, to be found in \cite[Theorem 3.3]{bish1}.}.  Note that the latter equivalence is meant in the sense of hyperimplication instead of implication.

\medskip

It suffices to prove $\paai$, by Theorem \ref{firstlast2}.
Let $\varphi$ be in $\Delta_{0}$ and assume $(\forall n\in\N)\varphi(n)$.  Let $w_{n}$ be an increasing sequence of reals which is bounded above by $B\in\N$  and convergent to $w\in \R$.
Without loss of generality, we may assume that the formulas $(\forall n\in \N)(w_{n}\leq w_{n+1})$ and  $(\exists N\in \N)(\forall n\in \N)( w_{n} \leq N )$ are in $\T$.
Note that the predicate `$\leq$' in the previous formulas does not increase their quantifier complexity (resp.\ $\Pi_{1}$ and $\Sigma_{2}$).
We define
\[
z_{n}:=
\begin{cases}
w_{n} & \text{if } (\forall m\leq n)\varphi(n) \\
w+\sum_{i=1}^{n}\frac{B+1-w}{2^{i}} & \text{otherwise}
\end{cases}.
\]
Then, $z_{n}$ is bounded and increasing, by assumption.  Also, by the definition of $z_{n}$, the formulas $(\forall n\in \N)(z_{n}\leq z_{n+1})$ and  $(\forall n\in \N)( z_{n} \leq B+1 )$ are in $\T$.
By $\mathbb{MCT}$, we have the consequent of \eqref{BIGGER2}, and the fact that this formula is in $\T$.  This yields
\[\textstyle
(\forall k\in \sN)(\exists M\in \sN\wedge \text{$M$ is $k$-finite })(\forall m,m'\in\sN)(m,m'\geq M\DI |z_{m}-z_{m'}| \leq \frac{1}{k}),
\]
by definition.  We tacitly assume that the universal quantifier hidden in the final predicate `$\leq$' of the previous formula has been adjusted in the right way.
Hence,
\be\label{tock2}\textstyle
(\forall k\in \N)(\exists M\in \N)(\forall m,m'\in\sN)(m,m'\geq M\DI |z_{m}-z_{m'}| \leq \frac{1}{k}),
\ee
Now if $\neg\varphi(n_{0})$ for some $n_{0}\in\sN$, then $z_{n_{0}}> w+\frac{1}{2}$.  However, applying \eqref{tock2} for $k=\frac{1}{2}$, there is some $M_{2}\in \N$ such that
$z_{M_{2}}~(=w_{M_{2}}\leq w)$ and $z_{n_{0}}$ satisfy $|z_{n_{0}}-z_{M_{2}}|\leq \frac{1}{2}$, which is a contradiction.
Hence, we must have $\varphi(n)$, for \emph{all} $n\in\sN$ if $(\forall n\in\N)\varphi(n)$, which yields $\paai$.

\medskip

For the forward implication, the usual `interval halving technique' can be easily adapted to derive $\MCTO$ from $\paai$.  Indeed, by the latter, a $\Sigma_{1}$-formula $\Phi(\vx)$ is equivalent to an $\Omega$-invariant formula, and $\Omega$-CA provides a function $f(\vx)$ such that $f(\vx)=1\asa \Phi(\vx)$, for standard parameters $\vx$.  Hence, inequalities between real numbers, which are $\Sigma_{1}$ in general, become decidable and the usual interval halving argument can be performed in $\NSA+\paai$.
\end{proof}
\begin{rem}\rm As it happens, formula \eqref{tock2} has a constructive interpretation akin to the definition of Cauchy sequence in BISH.
Indeed, Bishop defines a Cauchy sequence by the usual $\Pi_{3}$-definition
\be\label{durk}\textstyle
(\forall k\in \N)(\exists M\in \N)(\forall m,m'\in\N)(m,m'\geq M\di |z_{m}-z_{m'}| \leq \frac{1}{k}),
\ee
with the \emph{extra assumption} that there is some sequence $M_{k}$ witnessing the existential quantifier.  The latter sequence is usually called a \emph{modulus}.
Now, given \eqref{tock2}, the $\N\di \N$-sequence\footnote{As in the previous proof, we assume that the final predicate `$\leq$' has been adjusted in the right way.}
\[\textstyle
g(k)=(\mu M\leq \omega)(\forall m,m'\leq \omega)(m,m'\geq M\DI |z_{m}-z_{m'}|\leq \frac{1}{k}),
\]
is $\Omega$-invariant.  Using $\Omega$-CA, we obtain a \emph{standard} function $f$ such that $f(k)=g(k)$, for $k\in \N$.  This yields
\[\textstyle
(\forall k\in \N)(\forall m,m'\in\N)(m,m'\geq f(k)\DI |z_{m}-z_{m'}| \leq \frac{1}{k}),
\]
i.e., $f$ is a modulus for $z_{m}$.  As the latter formula is equivalent to \eqref{durk}, it is also in $\T$.
\end{rem}

\subsection{Markov's principle}
In this paragraph, we consider \emph{Markov's principle}.  This principle is rejected in Constructive Analysis, but accepted in the Markovian school of recursive mathematics (See \cite[p.\ 10-11]{bridges1}).
\begin{princ}[$\MPO$]
For $P$ in $\Sigma_{1}$, we have ${\sim}{\sim}P \DI P$.
\end{princ}
Note that $\MPO$ corresponds to double-negation elimination for $\Sigma_{1}$-formulas
\begin{princ}[$\mathbb{MPR}$]
$(\forall x\in \R )({\sim}{\sim}(x>0)\DI x>0)$.
\end{princ}
Let $(\paai)_{1}$ be $\Pi_{1}$-transfer limited to $\N_{1}$, i.e., in \eqref{xyz}, we replace $^{*}\N$ by $\N_{1}$.
We have the following theorem.
\begin{thm}\label{sproof}
In $\NSA$, $\MPO$ is equivalent to $\mathbb{MPR}$, and to $(\paai)_{1}$
\end{thm}
\begin{proof}
For $\varphi$ in $\Delta_{0}$, the formula ${\sim}{\sim}[(\exists n\in\N)\varphi(n)]$ is $(\exists n\in \sN)\varphi(n)$.  Hence, $\MPO$ for this formula yields
\be\label{kirk}
\big[(\exists n\in \sN)\varphi(n) \wedge [(\exists n\in \sN)\varphi(n)\di (\exists n\in \N_{1})\varphi(n)]\big] \di (\exists n\in \N)\varphi(n),
\ee
which reduces to $(\exists n\in \N_{1})\varphi(n)\di (\exists n\in \N)\varphi(n)$.  What is left to prove is that $\mathbb{MPR}$ implies $\MPO$.

\medskip

Hence, let $\varphi$ be $\Delta_{0}$ and assume  that ${\sim}{\sim}[(\exists n\in \N)\varphi(n)]$, i.e., $(\exists n\in \sN)\varphi(n)$, and that this formula is in $\T$.
By \eqref{kirk}, this yields $(\exists n\in \N_{1})\varphi(n)$.
Now define the real $x$ as $q_{n}:=\sum_{i=0}^{n}\frac{1}{2^{i}}T_{\psi}(i)$, for $\psi(i)$ defined as $(\exists m\leq i)\varphi(m)$, i.e., as in the proof of Theorem \ref{uni2}.
Note that we also have that `$x\in\R$' is in $\T$.
Now let $n_{1}\in \N_{1}$ be such that $\varphi(n_{1})$.  Then $q_{n}>q_{n_{1}}\geq \frac{1}{2^{n_{1}}}$, for all $n>n_{1}$, and we have $q_{n_{1}+1}>\frac{1}{2^{n_{1}+1}}$, i.e.,
$(\exists n\in \N_{1})(q_{n}>2^{n})$, which is easily seen to be ${\sim}{\sim}(x>0)$.   By $\mathbb{MPR}$, we have $x>0$, implying that $(\exists n\in \N)\varphi(n)$.
Hence, $\MPO$ follows and we are done.
\end{proof}
By the previous theorem, $\MPO$ is not provable in $\NSA$, and weaker than $\LPOO$.
It is not difficult to produce models of $\NSA$ in which, e.g., $\MPO$ (and hence $\LPOO$) fails, or where $\MPO$ holds, but $\LPOO$ fails.
Hence, $\MPO$ is \emph{strictly} weaker than $\LPOO$.  Similar results exists for other principles, such as $\WLPOO$.

\medskip

Furthermore, it can be shown that $\mathbb{\Delta}_{1}$-formulas\footnote{A formula $\varphi$ is $\mathbb{\Delta}_{1}$ if there are $\varphi_{1}\in \Pi_{1}$ and $\varphi_{2}\in \Sigma_{1}$ such that $\varphi\ASA \varphi_{1}\ASA \varphi_{2}$.} are $\Omega$-invariant \emph{in the presence of} $\MPO$.  Hence, by the absence of $\MPO$ in $\NSA$, it would seem that not all $\mathbb{\Delta}_{1}$-formulas correspond to `finite procedures', if the latter is interpreted as $\Omega$-invariant procedures.  Similarly, not all recursive operations are finite procedures\footnote{An example is a Turing machine $M$ for which it is impossible that it does not hold.  As no algorithmic upper bound on the halting time of $M$ is given, we cannot conclude that $M$ halts in BISH.} in BISH as Markov's principle is rejected \cite[p.\ 10-11]{bridges1}.
Hence, we observe a deep similarity here between $\NSA$ and BISH.

\medskip

Finally, as is the case in Constructive Analysis, double-negation elimination for $\Pi_{1}$-formulas is provable in the base theory.
\begin{thm}
In $\NSA$, we have ${\sim}{\sim}R\DI R$ and ${\sim}{\sim}R\di R$, for all $R$ in $\Pi_{1}$.
\end{thm}
\begin{proof}
Let $\varphi$ be $\Delta_{0}$ and consider ${\sim}[(\forall n\in \N)\varphi(n)]$.  By definition, the latter
is just $(\exists n\in \sN)\neg\varphi(n)$.  Hence, ${\sim}{\sim}[(\forall n\in \N)\varphi(n)]$ is $(\exists n\in \sN)\neg\varphi(n)\DI 0=1$, i.e.,
\[
\big[(\exists n\in \sN)\neg\varphi(n) [(\exists n\in \sN)\neg\varphi(n)\di (\exists n\in \N_{1})\neg\varphi(n)]\big] \di 0=1,
\]
which yields
\[
(\forall n\in \sN)\varphi(n) \vee [(\exists n\in \sN)\neg\varphi(n)\wedge (\forall n\in \N_{1})\varphi(n)].
\]
The latter is just $(\forall n\in \N_{1})\varphi(n)$.  The theorem now follows easily.
\end{proof}
\subsection{Concluding remarks}
In the previous section, we obtained several equivalences in $\NSA$ which are very reminiscent of Constructive Reverse Mathematics.
Obviously, these results do not even scratch the surface, and we could prove many more equivalences, given more space (e.g., for LLPO, WLPO, WMP, MP$^{\vee}$, FAN$_{\Delta}$, \dots).
In particular, we could obtain most equivalences in Hajime Ishihara's survey paper \cite{ishi1} of Constructive Reverse Mathematics in the same way.  
The following theorem provides an example.
\begin{thm} In $\NSA$, the following are equivalent.
\begin{enumerate}
\item \textup{$\mathbb{LLPO}$:}  ${\sim}(P \wedge Q)\Rrightarrow{\sim} P \HD {\sim} Q$ \quad $(P,Q\in \Sigma_{1})$.  \label{kyo3}
\item \textup{$\mathbb{LLPR}$:} $(\forall x\in \R)[{\sim}(x>0)\HD {\sim}(x<0)]$. \label{kyo4}
\item \textup{$\mathbb{NIL}$:} $(\forall x,y\in \R)(xy=0\Rrightarrow x=0\HD y=0)$.
\item \textup{$\mathbb{CLO}$:} For all $x,y\in \R$ with ${\sim}(x<y)$, $\{x,y\}$ is a closed set.
\item \textup{$\mathbb{IVT}$:} a version of the intermediate value theorem.
\item \textup{$\mathbb{WEI}$:} a version of the Weierstra\ss~extremum theorem.
\end{enumerate}
 In $\NSA$, the following are equivalent.
\begin{enumerate}
\item \textup{$\mathbb{WLPO}$:}  $ ({\sim}{\sim}P)\HD P$ \quad $(P\in \Sigma_{1})$.
\item \textup{$\mathbb{WLPO}$:} $(\forall x\in \R)[{\sim}{\sim}(x>0)\HD (x>0)]$.
\item \textup{$\mathbb{DISC}$:} There exists a discontinuous function from $2^{\N}$ to $\N$.
\end{enumerate}
\end{thm}

%
%
\begin{rem}\label{reuniting}\rm
In 1999, a conference entitled \emph{Reuniting the antipodes} was organized in Venice \cite{venice}.
The goal of this meeting was to bring together the communities of Nonstandard Analysis and Constructive Analysis in order to discover common ground.
In \cite{jaap1}, the review of \cite{cross}, van Oosten states that little common ground had been found.
He suggests, however, one notable exception: Erik Palmgren, who indeed has been developing Nonstandard Analysis inside the constructive framework (See, e.g., \cite{palmdijk,palmboominvenetie}).
This paper's approach, although dual in nature, can also been as an attempt at `reuniting the antipodes'.  In particular, despite Bishop's rather negative remarks regarding Nonstandard Analysis \cite{schizo}, the results in this paper can be interpreted positively, namely as providing Bishop's notion of algorithm with a certain \emph{robustness}.
\end{rem}

\bibliographystyle{eptcs}
\bibliography{sanders}

\end{document}